\documentclass[conference]{IEEEtran}
\usepackage[letterpaper, left=0.76in,right=0.76in,top=0.76in,bottom=0.77in, footskip=.25in]{geometry}

\usepackage{comment}
\usepackage{cite}
\usepackage{graphicx}
\usepackage{epstopdf}
\usepackage{amsthm}
\usepackage[cmex10]{amsmath}
\usepackage{array}
\usepackage{algorithm}
\usepackage[noend]{algpseudocode}
\usepackage{float}
\usepackage{graphicx}
\usepackage{amssymb}
\usepackage{microtype}
\usepackage{color}
\usepackage{relsize}
\usepackage[font=footnotesize]{caption}
\usepackage{multirow}
\usepackage{balance}
\usepackage{amsmath}
\usepackage{hhline}
\usepackage[caption=false]{subfig}
\usepackage{enumitem}
\usepackage{flushend}

\newtheorem{theorem}{Theorem}

\newtheorem{definition}{Definition}
\newtheorem{remark}{Remark}
\newtheorem{example}{Example}

\newtheorem{Corollary}{Corollary}
\bstctlcite{IEEEexample:BSTcontrol}

\makeatletter
\def\BState{\State\hskip-\ALG@thistlm}
\makeatother

\def\defop{\mathrel{\ooalign{%
\raisebox{1.2\height}{\scriptsize3}\cr\hidewidth$\neq$\hidewidth\cr}}}

\def\defopp{\mathrel{\ooalign{%
\raisebox{2\height}{\scriptsize p}\cr\hidewidth$\neq$\hidewidth\cr}}}

\begin{document}

\title{\vspace{0.25in}Multi-Dimensional Spatially-Coupled Code Design Through {Informed} Relocation of Circulants\vspace{-0.1cm}}

\author{\IEEEauthorblockN{Homa Esfahanizadeh, Ahmed Hareedy, and Lara Dolecek}
\IEEEauthorblockA{Department of Electrical and Computer Engineering, University of California, Los Angeles, USA\\
hesfahanizadeh@ucla.edu, ahareedy@ucla.edu, and dolecek@ee.ucla.edu}\vspace{-0.5cm}}

\maketitle

\begin{abstract}
A circulant-based spatially-coupled (SC) code is constructed by partitioning the circulants of an underlying block code into a number of components, and then {coupling copies of these components} together. By connecting (coupling) several SC codes, multi-dimensional SC (MD-SC) codes are constructed. In this paper, we present a systematic framework for constructing MD-SC codes with notably better girth properties than their 1D-SC counterparts. In our framework, informed multi-dimensional coupling is performed {via an optimal relocation and an (optional) power adjustment of problematic circulants in the constituent SC codes}. Compared to the 1D-SC codes, our MD-SC codes are demonstrated to have {up to 85\%} reduction in the population of the smallest cycle, and {up to 3.8 orders of magnitude BER improvement} in the early error floor region. {The results of this work can be particularly beneficial in data storage systems, e.g., 2D magnetic recording and 3D Flash systems, as high-performance MD-SC codes are robust against various channel impairments and non-uniformity.} 
\vspace{-0.0cm}
\end{abstract}

\IEEEpeerreviewmaketitle

\section{Introduction}

Spatially-coupled (SC) codes are a family of graph-based codes that have attracted significant attention thanks to their capacity approaching performance. SC codes are constructed by coupling together a series of disjoint block codes into a single coupled chain \cite{FelstromIT1999}. Here, we use circulant-based (CB) LDPC codes, \cite{TannerIT2004}, as the underlying block codes. Multi-dimensional SC (MD-SC) codes are constructed by coupling several SC codes together via rewiring the existing connections or by adding extra variable nodes (VNs) or check nodes (CNs). MD-SC codes are more robust against burst erasures and channel non-uniformity, and they have improved iterative decoding thresholds, compared to 1D-SC codes {\cite{TruhachevITA2012,OhashiISIT2013}}.

In \cite{TruhachevITA2012}, {a construction method is presented for MD-SC codes that have specific structures, e.g., {loops and triangles}.}
The construction method for MD-SC codes presented in \cite{OhashiISIT2013} involves connecting edges uniformly at random such that some criteria on the number of connections are satisfied.
In \cite{SchmalenISTC2014}, a framework is presented for constructing MD-SC codes by randomly and sparsely introducing additional VNs {to connect CNs at the same positions of different chains}. 
In \cite{LiuCOMML2015}, multiple SC codes are connected by edge exchange between adjacent chains to improve the iterative decoding threshold. 
{The previous works on MD-SC codes, while promising, have some limitations. In particular, they either consider random constructions or are limited to specific topologies. They also use the density evolution technique for the performance analysis. This technique is dedicated to the asymptotic regime and is based on some assumptions that can not be readily translated to the practical finite-length case.} In \cite{OlmosITW2013}, a finite-length analysis in the waterfall region for MD-SC codes with a loop structure is presented.


Finding the best connections to be rewired in order to connect constituent 1D-SC codes and construct MD-SC codes with high finite-length performance is still an open problem. {This is the first paper to present a systematic framework for constructing MD-SC codes by optimally coupling individual SC codes together to attain {fewer} short cycles.} For connecting the constituent SC codes, we do not add extra VNs or CNs, and we only rewire some existing connections. 

For exchanging the connections, we follow three rules: (1) {The connections involved in the highest number of {certain} cycles are targeted for rewiring{;}} (2) {The neighboring {constituent 1D-}SC code to which the targeted connections are rewired is chosen such that the minimum number of {the certain} cycles is attained{;}} (3) The targeted connections are rewired to the same positions in the other {constituent} 1D-SC codes in order to preserve the low-latency decoding property. From an algebraic viewpoint, problematic circulants {(which correspond to groups of connections)} that contribute to the highest number of {certain} cycles in the {constituent 1D-SC} codes are relocated {to connect these codes {together}.} Finally, the powers of relocated circulants are (optionally) adjusted to further improve {girth properties}.


\section{Preliminaries}

Throughout this paper, each column (resp., row) in a parity-check matrix corresponds to a VN (resp., CN) in the equivalent graph of the matrix. CB codes are regular $(\gamma,\kappa)$ LDPC codes, where $\gamma$ is the column weight of the parity-check matrix (VN degree), and $\kappa$ is the row weight (CN degree). The parity-check matrix $\bold{H}$ of a CB code consists of $\kappa\gamma$ circulants. Each circulant is of the form of $\sigma^{f_{i,j}}$ where $i$, $0\leq i\leq \gamma {-}1$, is the row group index, $j$, $0\leq j\leq \kappa {-} 1$, is the column group index, and $\sigma$ is the $z\times z$ identity matrix cyclically shifted one unit to the left (a circulant permutation matrix). In this paper, we use CB codes as the underlying block codes to construct SC codes

The parity-check matrix $\bold{H}_\text{SC}$ of a CB SC code is constructed by partitioning the $\kappa\gamma$ circulants of the underlying block code into ($m+1$) component matrices $\bold{H}_0,\bold{H}_1,\dots,\bold{H}_m$ (with the same size as $\bold{H}$), and piecing $L$ copies of the component matrices together as shown in Fig.~1. The parameters $m$ and $L$ are called the memory and the coupling length, respectively. Each component matrix $\bold{H}_l$, $0\leq l \leq m$, has a subset of circulants of $\bold{H}$ and {zeros} elsewhere {so that} $\sum_{l=0}^{m}\bold{H}_l=\bold{H}$. A replica $\bold{R}_d$, $1\leq d \leq L$, is a submatrix of $\bold{H}_\text{SC}$ that has one non-zero submatrix $[\bold{H}_0^T\dots\bold{H}_m^T]^T$, see Fig.~1. Recently, a systematic framework for partitioning the underlying block code and optimizing the circulant powers, known as the OO-CPO technique, was {proposed} for constructing high-performance SC codes \cite{1802_06481}. In this paper, we use the OO-CPO technique for designing the constituent SC codes that are {then} used to construct MD-SC codes. 

{Short cycles have a negative impact on the performance of LDPC codes under iterative decoding}. These cycles affect the independence of the extrinsic information exchanged in the iterative decoder. Moreover, problematic combinatorial objects that cause the error-floor phenomenon, e.g., absorbing sets and trapping sets, are formed of cycles with relatively short lengths in the graph of a code \cite{Richardson2003,DolecekIT2010}. {We present a systematic framework to construct MD-SC codes, which is based on an optimal relocation of circulants. MD-SC codes constructed using our proposed framework enjoy notably lower population of short cycles, and consequently better performance compared to {1D-SC codes.}}

\begin{figure}
\centering
\includegraphics[width=0.26\textwidth]{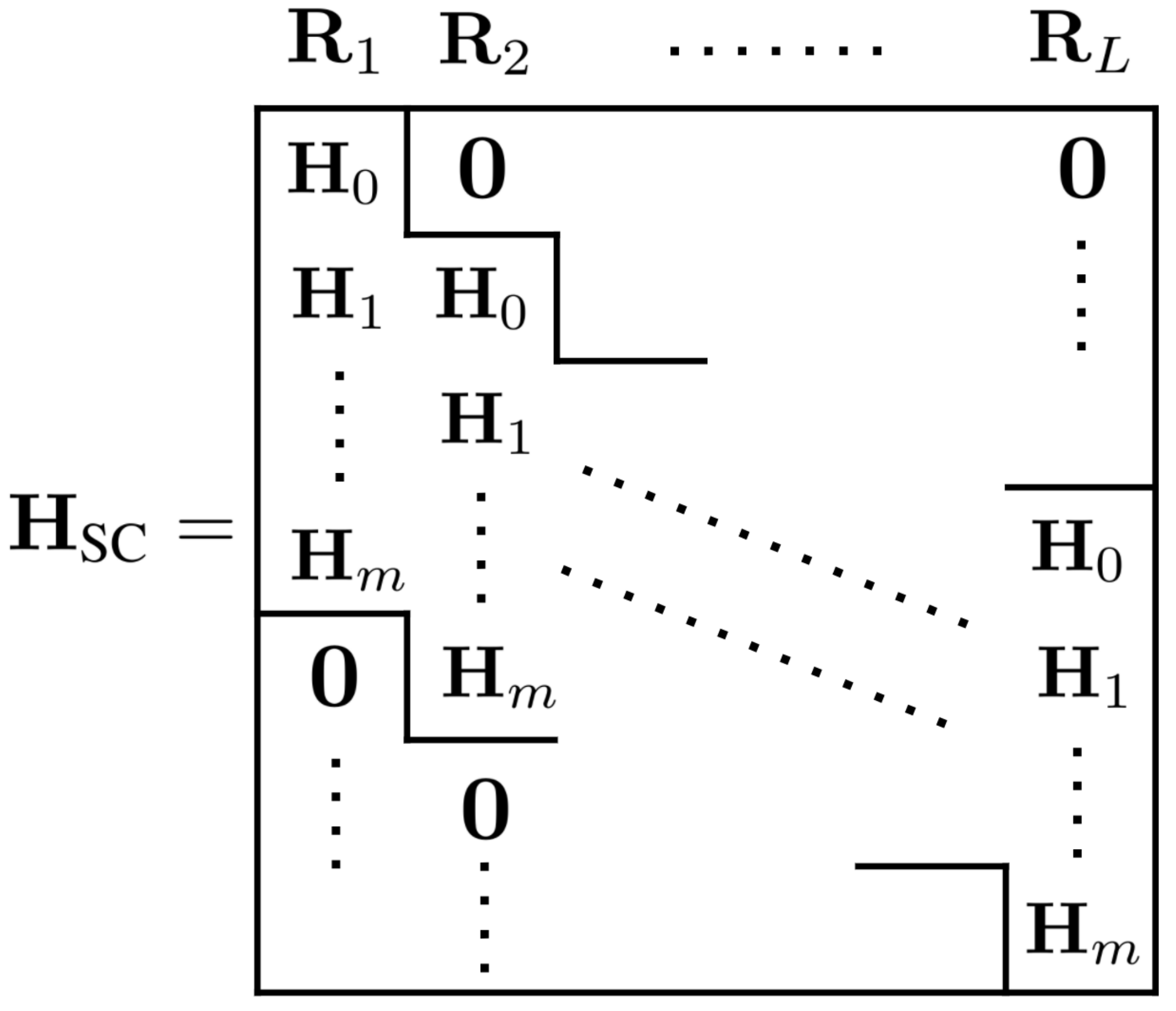}\vspace{-0.0cm}
\caption{The parity-check matrix of an SC code with parameters $m$ and $L$.\vspace{-0.5cm}}
\end{figure}

\section{Novel Framework for MD-SC Code Design}
First, we present our system model for MD-SC codes. Next, we investigate the effects of relocating a subset of circulants on the population of cycles. Finally, we present our algorithm for constructing MD-SC codes based on a majority voting policy.
\subsection{MD-SC Code Structure}
Consider three instances of an SC code with {parity-check matrix} $\bold{H}_\text{SC}$, memory $m$, and coupling length $L$. Consider the middle replica {$\bold{R}_d$} in $\bold{H}_\text{SC}${, where $d=\lceil L/2 \rceil$}. There are $\kappa\gamma$ non-zero circulants in this replica. Out of these $\kappa\gamma$ circulants, we choose $\mathcal{T}$ circulants that are the most problematic, i.e., that contribute to the highest number of cycles-$k$, where $k$ is the girth of the SC code. {We relocate the chosen circulants to two auxiliary matrices, $\bold{P}$ and $\bold{Q}$, such that a relocated circulant from $\bold{H}_{\text{SC}}$ is moved to the same position in either $\bold{P}$ or $\bold{Q}$}. The same relocations are repeated for all the ($L-1$) remaining replicas.
{We note that the middle replica $\bold{R}_d$ is considered for ranking the circulants in order to include all possible cycles-$k$ that a non-zero circulant in $\bold{H}_\text{SC}$ can contribute to.}
Initially, $\bold{P}$ and $\bold{Q}$ are set to zero. These auxiliary matrices have the same dimensions as $\bold{H}_\text{SC}$, and
\begin{equation}
\bold{H}_\text{SC}=\bold{H}_\text{SC}'+\bold{P}+\bold{Q},
\end{equation}
where $\bold{H}_\text{SC}'$ is derived from $\bold{H}_\text{SC}$ by removing the $\mathcal{T}$ chosen circulants.
{The parity-check matrix of the MD-SC code, $\bold{H}_\text{SC}^\text{MD}$, is then constructed as follows:}
\begin{equation}
\bold{H}_\text{SC}^\text{MD}=\left[\begin{array}{lll}
\bold{H}_\text{SC}'&\bold{Q}&\bold{P}\\
\bold{P}&\bold{H}_\text{SC}'&\bold{Q}\\
\bold{Q}&\bold{P}&\bold{H}_\text{SC}'
\end{array}\right].
\end{equation}

Example~1 shows the graphical illustration of an MD-SC code having the presented structure.
\begin{example}
Consider an SC code with $\kappa=3$, $\gamma=2$, $z=3$, $m=1$, and $L=3$. {The matrix $\bold{H}$ of the underlying block code and the component matrices are given below:
\begin{equation*}
\bold{H}=\left[\begin{array}{ccc}
\sigma^{f_{0,0}}&\sigma^{f_{0,1}}&\sigma^{f_{0,2}}\\
\sigma^{f_{1,0}}&\sigma^{f_{1,1}}&\sigma^{f_{1,2}}
\end{array}\right],\\
\end{equation*}
\begin{equation*}
\bold{H}_0\hspace{-0.11cm}=\hspace{-0.12cm}\left[\begin{array}{ccc}
\sigma^{f_{0,0}}&\bold{0}&\sigma^{f_{0,2}}\\
\bold{0}&\sigma^{f_{1,1}}&\bold{0}
\end{array}\right]\hspace{-0.1cm},
\bold{H}_1\hspace{-0.11cm}=\hspace{-0.12cm}\left[\begin{array}{ccc}
\bold{0}&\sigma^{f_{0,1}}&\bold{0}\\
\sigma^{f_{1,0}}&\bold{0}&\sigma^{f_{1,2}}
\end{array}\right]\hspace{-0.1cm}.
\end{equation*}
}Three instances of the code along with problematic connections are depicted in Fig.~2(a). The problematic connections of an SC code are rewired to the same positions in another SC code to construct an MD-SC code with $\mathcal{T}=1$, as shown in Fig.~2(b). {
Here, the $L$ instances of circulant $\sigma^{f_{1,0}}$ in each $\bold{H}_\textnormal{SC}$ are relocated to construct the MD-SC code.
}
\begin{figure}
\centering
\begin{tabular}{c}
\includegraphics[width=0.4\textwidth]{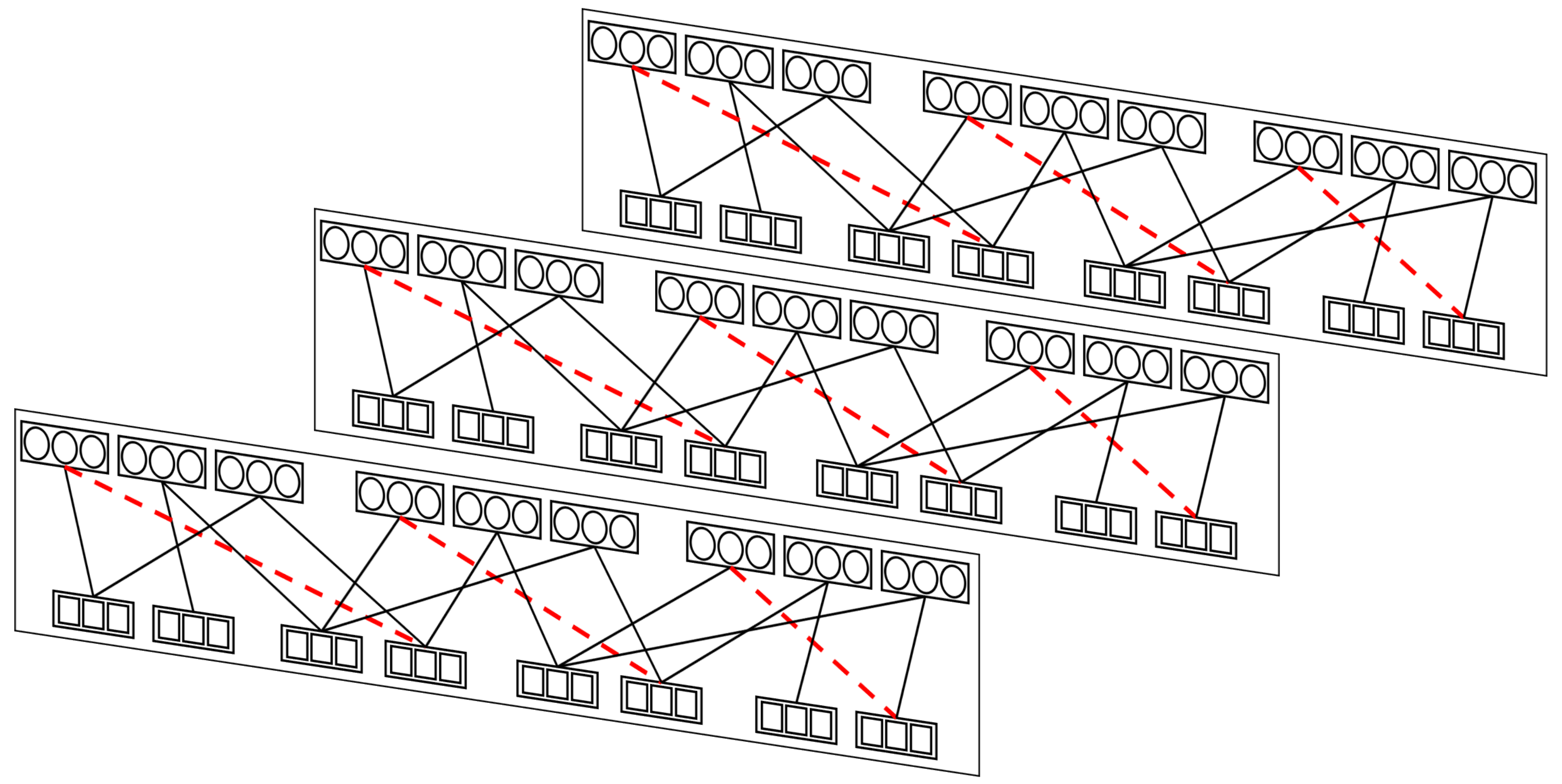}\\
(a)\vspace{-0.1cm}\\
\includegraphics[width=0.4\textwidth]{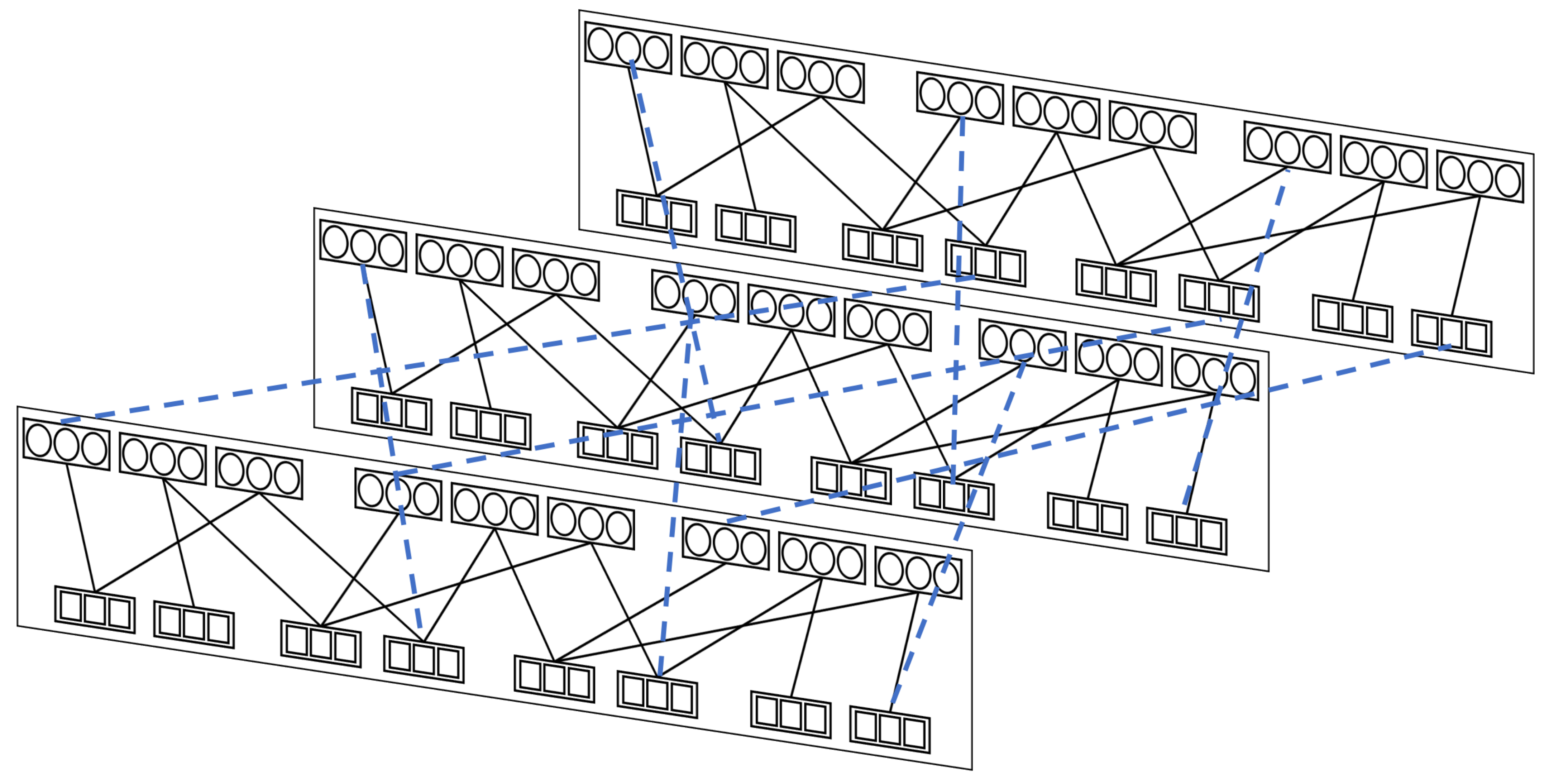}\\
(b)
\end{tabular}\vspace{-0.1cm}
\caption{(a) Three {1D-SC} codes. Circles (resp., squares) represent VNs (resp., CNs). Each line represents a group of connections (defined by a circulant) from $z$ VNs to $z$ CNs. Problematic connections are shown in {dashed} red lines. (b) MD-SC code. Rewired connections are shown in dashed blue lines.\vspace{-0.5cm}}
\end{figure}
\end{example}
\begin{definition}\text{ }
\begin{enumerate}
\item Let $\mathcal{C}_{i,j}$, where $0\leq i\leq (L+m)\gamma {-}1$ and $0\leq j\leq L\kappa {-} 1$, be a non-zero circulant in $\bold{H}_\textnormal{SC}$. We say $\mathcal{C}_{i,j}$ is relocated to $\bold{P}$ (resp., $\bold{Q}$) if it is moved from $\bold{H}_\textnormal{SC}$ to $\bold{P}$ (resp., $\bold{Q}$). We denote this relocation as $\mathcal{C}_{i,j}{\rightarrow}\bold{P}$ (resp., $\mathcal{C}_{i,j}{\rightarrow}\bold{Q}$).
\item {The operator $\overset{p}{=}$ (resp., $\defopp$\hspace{0.05cm}) defines the congruence (resp., incongruence) modulo $p$, and the operator $(.)_{p}$ defines modulo $p$ of an integer.}
\item The MD mapping $M:\{\mathcal{C}_{i,j}\}{\rightarrow}\{0,1,2\}$ is a mapping from a non-zero circulant in $\bold{H}_\textnormal{SC}$ to an integer in $\{0,1,2\}$, and it is defined as follows:
\begin{enumerate}
\item If $\mathcal{C}_{i,j}{\rightarrow}\bold{P}$, $M(\mathcal{C}_{i,j})=1$.
\item If $\mathcal{C}_{i,j}{\rightarrow}\bold{Q}$, $M(\mathcal{C}_{i,j})=2$. 
\item {If $\mathcal{C}_{i,j}$ is kept in $\bold{H}_\textnormal{SC}'$ (no relocation), $M(\mathcal{C}_{i,j})=0$.}
\end{enumerate}
\item A cycle-$k$, or $\mathcal{O}_k$, visits $k$ circulants in the parity-check matrix of the code, see Fig.~3. We list the $k$ circulants of $\mathcal{O}_k$, according to the order they are visited when the cycle is traversed in a clockwise direction, in a sequence as $C_{\mathcal{O}_k}=\{\mathcal{C}_{i_1,j_1},\mathcal{C}_{i_2,j_2},\dots,\mathcal{C}_{i_k,j_k}\}$, where $i_1=i_2,j_2=j_3,\dots,i_{k-1}=i_k,j_k=j_1$. A circulant can be visited more than once by $\mathcal{O}_k$, e.g., see Fig.~3(b).
\item We denote the distance between two circulants $\mathcal{C}_{i_u,j_u}$ and $\mathcal{C}_{i_v,j_v}$ on a cycle $\mathcal{O}_k$, where $u,v\in\{1,\dots,k\}$, as $D_{\mathcal{O}_k}(\mathcal{C}_{i_u,j_u},\mathcal{C}_{i_v,j_v})\in\{0,\dots,k-1\}$.
By definition, $D_{\mathcal{O}_k}(\mathcal{C}_{i_u,j_u},\mathcal{C}_{i_v,j_v})=|v-u|$.
\end{enumerate}
\end{definition}
Because of the structure of MD-SC codes, when a non-zero circulant in one replica of $\bold{H}_\textnormal{SC}$ is relocated, the same relocation is applied to the ($L-1$) other replicas as well. Thus,
\begin{equation}
{M(\mathcal{C}_{i,j})=M(\mathcal{C}_{i-\rho \gamma,j-\rho \kappa}), \text{ where } \rho=\lfloor j/\kappa \rfloor.}
\end{equation}

\begin{figure}
\centering
\begin{tabular}{m{4cm}m{4cm}}
\includegraphics[width=0.17\textwidth]{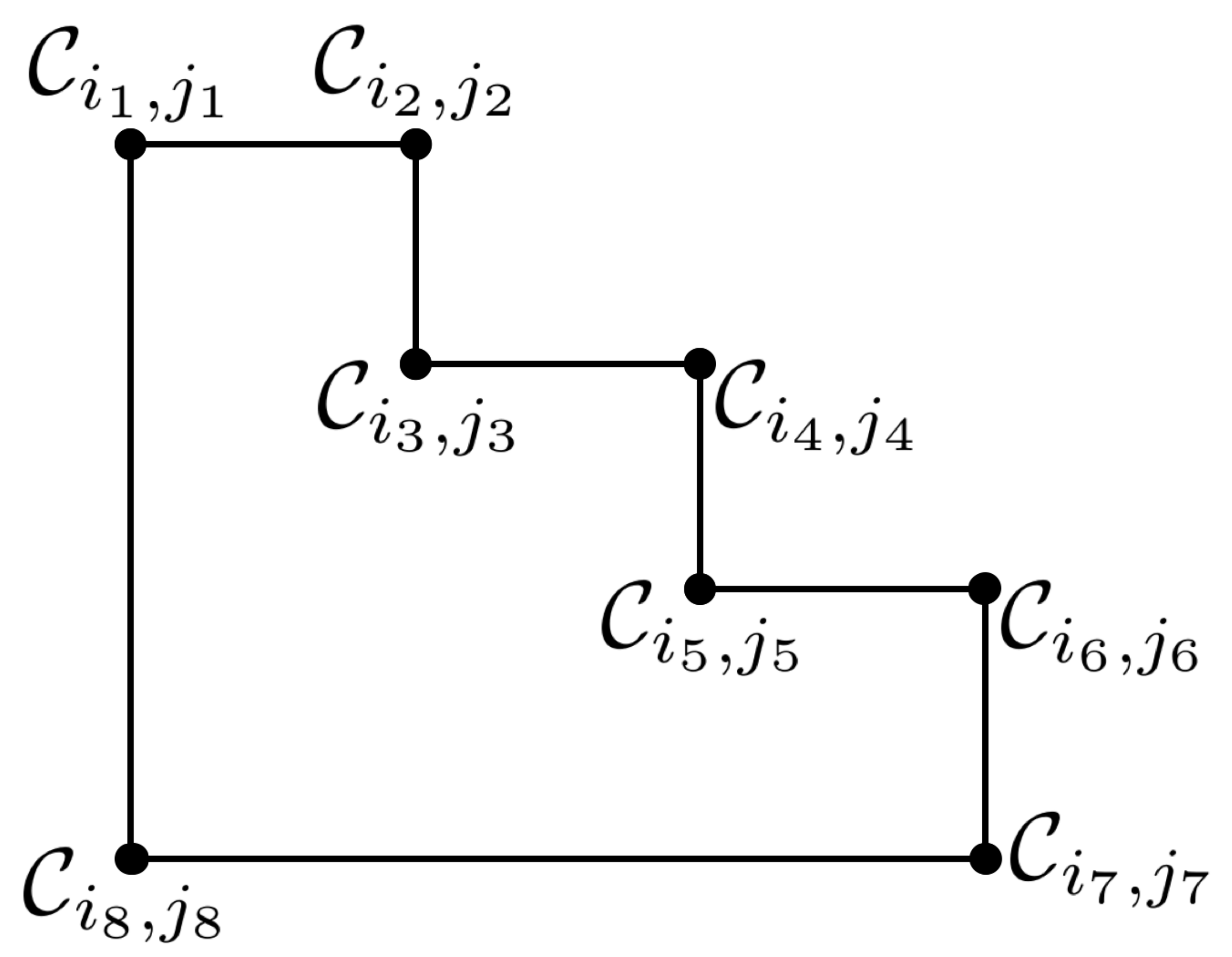}&
\includegraphics[width=0.17\textwidth]{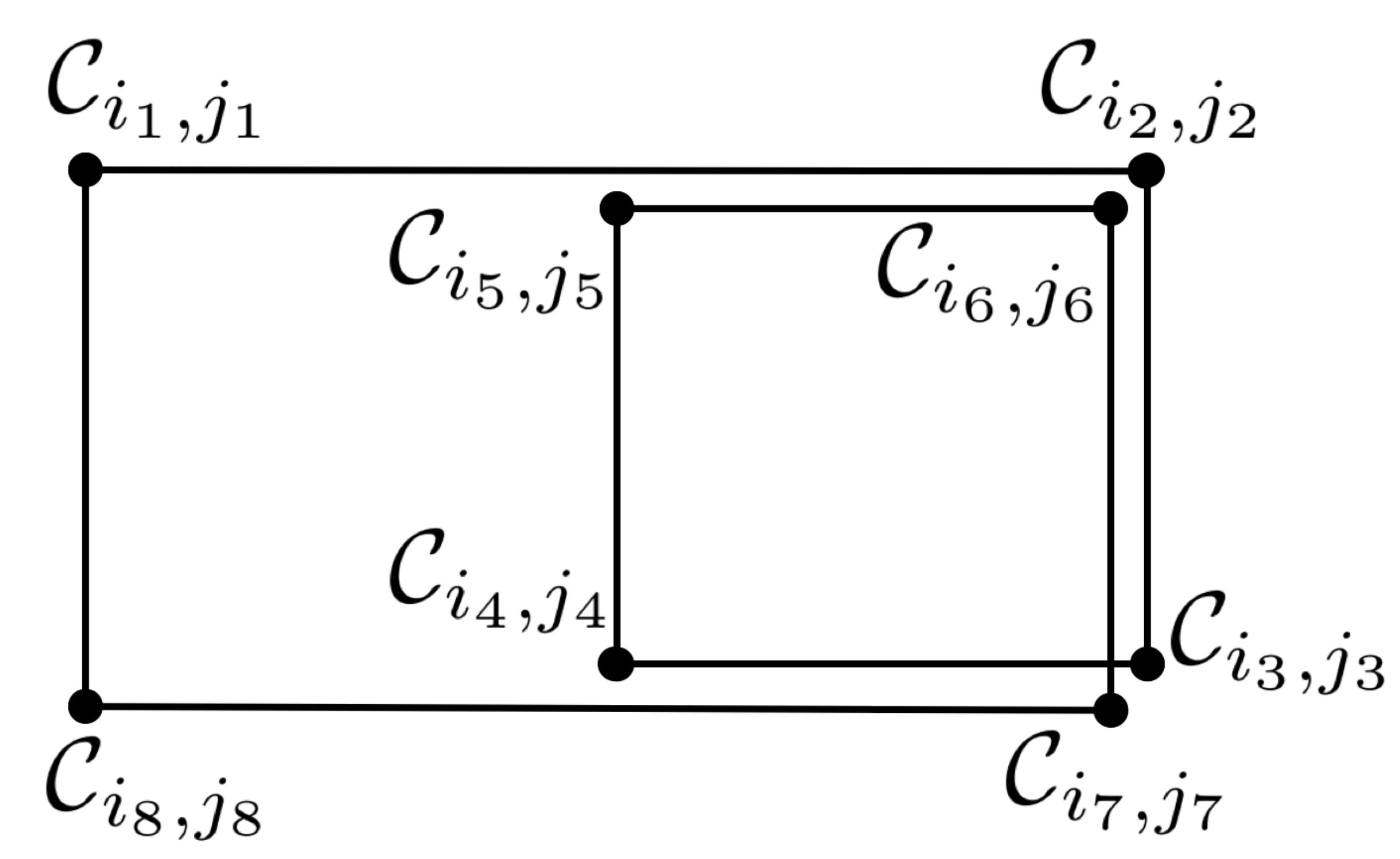}\vspace{-0.5cm}\\
\hspace{1.6cm}(a)&\hspace{1.6cm}(b)
\end{tabular}\vspace{-0.1cm}
\caption{Cycles-$8$ with $C_{\mathcal{O}_8}=\{\mathcal{C}_{i_1,j_1},\dots,\mathcal{C}_{i_8,j_8}\}$. Each line represents a connection between two circulants. (a) All circulants are unique. (b) $\mathcal{C}_{i_6,j_6}=\mathcal{C}_{i_2,j_2}$ and $\mathcal{C}_{i_7,j_7}=\mathcal{C}_{i_3,j_3}$.\vspace{-0.5cm}}
\end{figure}

{In the new MD-SC code design framework, we effectively answer two questions: which circulants to relocate, and where to relocate them.} We note that the relocations of circulants to the same positions in the auxiliary matrices preserve the special structure of SC codes, which makes them suitable for  applications that require low decoding latency.

\subsection{The Effects of Relocation of Circulants on Cycles}
We investigate the effect of relocating a subset of circulants of cycle $\mathcal{O}_k$, and we call this subset \textit{targeted} circulants. As we show, some relocations remove $\mathcal{O}_k$ from $\bold{H}_\text{SC}^\text{MD}$, while others preserve the cycle. We call the former group \textit{effective} relocations and the latter \textit{ineffective} relocations.\vspace{-0.0cm}

\begin{theorem}
Let $C_{\mathcal{O}_k}=\{\mathcal{C}_{i_1,j_1},\mathcal{C}_{i_2,j_2},\dots,\mathcal{C}_{i_k,j_k}\}$ be the sequence of circulants in $\bold{H}_\textnormal{SC}$ that are visited in a clockwise order by $\mathcal{O}_k$. If the following equation holds, the cycle-$k$ is preserved in $\bold{H}_\textnormal{SC}^\textnormal{MD}$ (ineffective relocations),
\vspace{-0.0cm}\begin{equation}
\sum_{u=1}^{k}(-1)^{u}M(\mathcal{C}_{i_u,j_u})\overset{3}{=}0.
\vspace{-0.00cm}\end{equation}
Otherwise, the relocations are effective, and the cycle $\mathcal{O}_k$ is removed from the MD-SC code
\footnote{{Equation (4) resembles Fossorier's condition on circulant powers of a CB code that makes a cycle in the protograph results in multiple cycles in the lifted graph of the code \cite{FossorierIT2004}.}}.
\end{theorem}

\begin{proof}
Let $(\mathcal{C}_{i_u,j_u},\mathcal{C}_{i_{u+1},j_{u+1}})$ be a pair of consecutive circulants on $\mathcal{O}_k$, where $u\in\{1,\dots,k\}$ and $\mathcal{C}_{i_{k+1},j_{k+1}}=\mathcal{C}_{i_1,j_1}$. By definition, two circulants have the same row (resp., column) group index, i.e., $i_u=i_{u+1}$ (resp., $j_u=j_{u+1}$), when $u\overset{2}{=}1$ (resp., $u\overset{2}{=}0$).

The matrix $\bold{H}_\textnormal{SC}^\textnormal{MD}$ is formed of $9$ submatrices, see (2).
{Here, a unit of {a} multi-dimensional horizontal (resp., vertical) shift is defined as cyclically going one submatrix right (resp., down) when {we go from $\mathcal{C}_{i_u,j_u}$ to $\mathcal{C}_{i_{u+1},j_{u+1}}$.}} {The submatrices of $\bold{H}_\textnormal{SC}^\textnormal{MD}$ appear in the cyclic order $\{\bold{H}_\textnormal{SC}',\bold{Q},\bold{P}\}$, with the MD mapping $\{0,2,1\}$, {from left to right}. These submatrices appear in the cyclic order $\{\bold{H}_\textnormal{SC}',\bold{P},\bold{Q}\}$, with the MD mapping $\{0,1,2\}$, {from top to bottom}. Thus, the multi-dimensional horizontal shift, when we go from $\mathcal{C}_{i_u,j_u}$ to $\mathcal{C}_{i_{u+1},j_{u+1}}$, $u\in\{1,3,\dots,k-1\}$, is $(M(\mathcal{C}_{i_u,j_u})-M(\mathcal{C}_{i_{u+1},j_{u+1}}))_{3}$ units.}
Similarly, the multi-dimensional vertical shift, when we go from $\mathcal{C}_{i_u,j_u}$ to $\mathcal{C}_{i_{u+1},j_{u+1}}$, $u\in\{0,2,\dots,k\}$, is $(M(\mathcal{C}_{i_{u+1},j_{u+1}})-M(\mathcal{C}_{i_u,j_u}))_{3}$ units. The total multi-dimensional horizontal and vertical shifts when we traverse the circulants of $\mathcal{O}_k$ in $\bold{H}_\textnormal{SC}^\textnormal{MD}$ are $\delta_{H}$ and $\delta_{V}$, respectively:
\begin{equation}
\begin{split}
&\delta_{H}=(\hspace{-0.5cm}\sum_{u\in\{1,3\dots,k-1\}}\hspace{-0.5cm}[M(\mathcal{C}_{i_u,j_u})-M(\mathcal{C}_{i_{u+1},j_{u+1}})])_{3}\\
&=(\sum_{u=1}^{k}[(-1)^{u+1}M(\mathcal{C}_{i_u,j_u})])_{3}\hspace{-0.07cm}=\hspace{-0.07cm}(-\sum_{u=1}^{k}[(-1)^{u}M(\mathcal{C}_{i_u,j_u})])_{3}{,}\\
&\delta_{V}=(\hspace{-0.34cm}\sum_{u\in\{2,4\dots,k\}}\hspace{-0.34cm}[M(\mathcal{C}_{i_{u+1},j_{u+1}})-M(\mathcal{C}_{i_u,j_u})])_{3}\\
&=(\sum_{u=1}^{k}[(-1)^{u+1}M(\mathcal{C}_{i_u,j_u})])_{3}\hspace{-0.07cm}=\hspace{-0.07cm}(-\sum_{u=1}^{k}[(-1)^{u}M(\mathcal{C}_{i_u,j_u})])_{3}.\\
\end{split}
\vspace{-0cm}
\end{equation}

The cycle $\mathcal{O}_k$ exists ({is} preserved) in $\bold{H}_\textnormal{SC}^\textnormal{MD}$, {if and only if} the start and end submatrices are the same when we traverse the $k$ circulants of $\mathcal{O}_k$. This requires the total multi-dimensional horizontal and vertical shifts ($\delta_{H}$ and $\delta_{V}$) to be zero, which results in (4). {Otherwise, the cycle is removed.} 
\end{proof}

\begin{Corollary}
Let  $\{\mathcal{C}_{i_x,j_x} \text{ } | \text{ }\mathcal{C}_{i_x,j_x}\in C_{\mathcal{O}_k},M(\mathcal{C}_{i_x,j_x})>0\}$ be the set of relocated circulants of $\mathcal{O}_k$ with cardinality $n$. The relocations are ineffective {if and only if} the following equation holds for the relocated circulants:
\begin{equation}
\sum_{x\overset{2}{=}0}M(\mathcal{C}_{i_x,j_x})\overset{3}{=}\sum_{x\overset{2}{=}1}M(\mathcal{C}_{i_x,j_x}).
\end{equation}
\end{Corollary}
If equation (4) holds for the circulants of $\mathcal{O}_k$, three instances of the cycle, in three constituent SC codes, form three cycles-$k$ in the MD-SC code (ineffective relocations). Theorem~2 investigates the situation when (4) is not satisfied.
\begin{theorem}
If (4) does not hold for the circulants of a cycle-$k$, three instances of the cycle, in three constituent SC codes, form a cycle of length $3k$ in the MD-SC code.
\end{theorem}
\begin{proof}
Consider a cycle $\mathcal{O}_k$ with $C_{\mathcal{O}_k}=\{\mathcal{C}_{i_1,j_1},\dots,\mathcal{C}_{i_k,j_k}\}$ such that $(-\sum_{u=1}^{k}[(-1)^{u}M(\mathcal{C}_{i_u,j_u})])_{3}=y$, where $y=1$ or $2$. {Consequently, (4) does not hold, and the cycle-$k$ is not preserved in the MD-SC code.} We traverse the circulants of $\mathcal{O}_k$ starting from $\mathcal{C}_{i_1,j_1}$ in a clockwise order. After traversing all $k$ circulants, we reach circulant $\mathcal{C}_{i_1,j_1}$ in {a submatrix that is (cyclically) $y\defop0$ units right and $y\defop0$ units down from the submatrix we started from}. {Thus, a cycle of length $k$ cannot be completed.} We proceed traversing the circulants until we reach $\mathcal{C}_{i_1,j_1}$ in a submatrix that is $2y\defop0$ units right and $2y\defop0$ units down from the submatrix we started from. {Thus, even a cycle of length $2k$ cannot be completed.} We proceed traversing the circulants one more time, and we reach circulant $\mathcal{C}_{i_1,j_1}$ in a submatrix that is $3y\overset{3}{=}0$ units right and $3y\overset{3}{=}0$ units down from the submatrix we started from, which is basically the submatrix we started traversing from. {Consequently, the resulting cycle is of length $3k$.}\vspace{-0.0cm}
\end{proof}
{By connecting three SC codes with girth $k$, according to the structure we presented in this paper, the girth of the MD-SC code is at least $k$. In particular, the three instances of each cycle-$k$ in the constituent SC codes result in either three cycles-$k$ or one cycle-$3k$ in the MD-SC code. The first situation is caused by ineffective relocations, and is avoided as much as possible in the algorithm that we present for the MD-SC code construction.}
In Example~2, we explore Theorems~1 and 2 for different scenarios. In each scenario, a different subset of circulants of $\mathcal{O}_k$ are relocated.\vspace{-0.0cm}
\begin{example} Let $C_{\mathcal{O}_k}=\{\mathcal{C}_{i_1,j_1},\dots,\mathcal{C}_{i_k,j_k}\}$ be the sequence of circulants of $\mathcal{O}_k$, and $n$ be number of its relocated circulants.
\begin{enumerate}[leftmargin=0.5cm]
\item Let $n=1$ and $\mathcal{C}_{i_a,j_a}$ be the relocated circulant. Then,  $\sum_{u=1}^{k}(-1)^{u}M(\mathcal{C}_{i_u,j_u})\defop0$. Therefore, $\mathcal{O}_k$ is removed, and a cycle-$3k$ is formed. Fig.~4(a) shows $\mathcal{C}_{i_a,j_a}{\rightarrow}\bold{P}$. Fig.~4(b) shows that a cycle-$3k$ (shown in orange) is formed because of the relocation. The green solid border represents that this relocation is effective.
\item {Let $n=2$ and $\mathcal{C}_{i_a,j_a}$, $\mathcal{C}_{i_b,j_b}$ be the relocated circulants.} Then, {whether} $\mathcal{O}_k$ is removed or not {depends} on the arrangement of the relocated circulants on $\mathcal{O}_k$ along with the auxiliary matrices they are relocated to. {Suppose} $D_{\mathcal{O}_k}(\mathcal{C}_{i_a,j_a},\mathcal{C}_{i_b,j_b})=1$, and the two circulantes are relocated to the same auxiliary matrix. {Thus,} $\sum_{u=1}^{k}(-1)^{u}M(\mathcal{C}_{i_u,j_u})\overset{3}{=}0$ and $\mathcal{O}_k$ is preserved (ineffective relocations). Fig.~5(a) shows $\mathcal{C}_{i_a,j_a},\mathcal{C}_{i_b,j_b}{\rightarrow}\bold{P}$. Fig.~5(b) shows that three cycles-$k$ (purple, blue, and yellow) are formed. The red solid border represents that these relocations are ineffective. 
\item If all $n$ targeted circulants are relocated to $\bold{P}$ (resp., $\bold{Q}$), and each pair of consecutive relocated circulants {has} even distance on $\mathcal{O}_k${, see Definition~1(5)}, the relocations are ineffective when $n\overset{3}{=}0$. Let $\mathcal{C}_{i_v,j_v}$ be the first relocated circulant in $C_{\mathcal{O}_k}$. {Then,}\vspace{-0.0cm}
\begin{equation}
\begin{split}
&\sum_{u=1}^{k}(-1)^{u}M(\mathcal{C}_{i_u,j_u})=\sum_{M(\mathcal{C}_{i_u,j_u})>0}(-1)^{u}M(\mathcal{C}_{i_u,j_u})\\
&\hspace{0.3cm}=M(\mathcal{C}_{i_v,j_v})\sum_{M(\mathcal{C}_{i_u,j_u})>0}(-1)^{u}\\
&\hspace{0.3cm}=M(\mathcal{C}_{i_v,j_v})(-1)^{v}(1+\dots+1)=nM(\mathcal{C}_{i_v,j_v})(-1)^{v},
\end{split}\vspace{-0.3cm}
\end{equation}
which $\overset{3}{=}0$ only if $n\overset{3}{=}0$.
\item {If all $n$ targeted circulants are relocated to $\bold{P}$ (resp., $\bold{Q}$), and each pair of consecutive relocated circulants {has} odd distance on $\mathcal{O}_k$,{, see Definition~1(5),} which can only happen if $n\overset{2}{=}0$ since $k\overset{2}{=}0$, the relocations are always ineffective.} Let $\mathcal{C}_{i_v,j_v}$ be the first relocated circulant in $C_{\mathcal{O}_k}$. {Then,}
\begin{equation}
\begin{split}
\sum_{u=1}^{k}(-1)^{u}&M(\mathcal{C}_{i_u,j_u})=M(\mathcal{C}_{i_v,j_v})\sum_{M(\mathcal{C}_{i_u,j_u})>0}(-1)^{u}\\
&=M(\mathcal{C}_{i_v,j_v})(-1)^{v}(1-1+\dots-1),
\end{split}
\end{equation}
{which $\overset{3}{=}0$ always since $n\overset{2}{=}0$.}
\end{enumerate}
\begin{figure}
\centering
\label{fig_1_circulant}
\begin{tabular}{cc}
\includegraphics[width=0.20\textwidth]{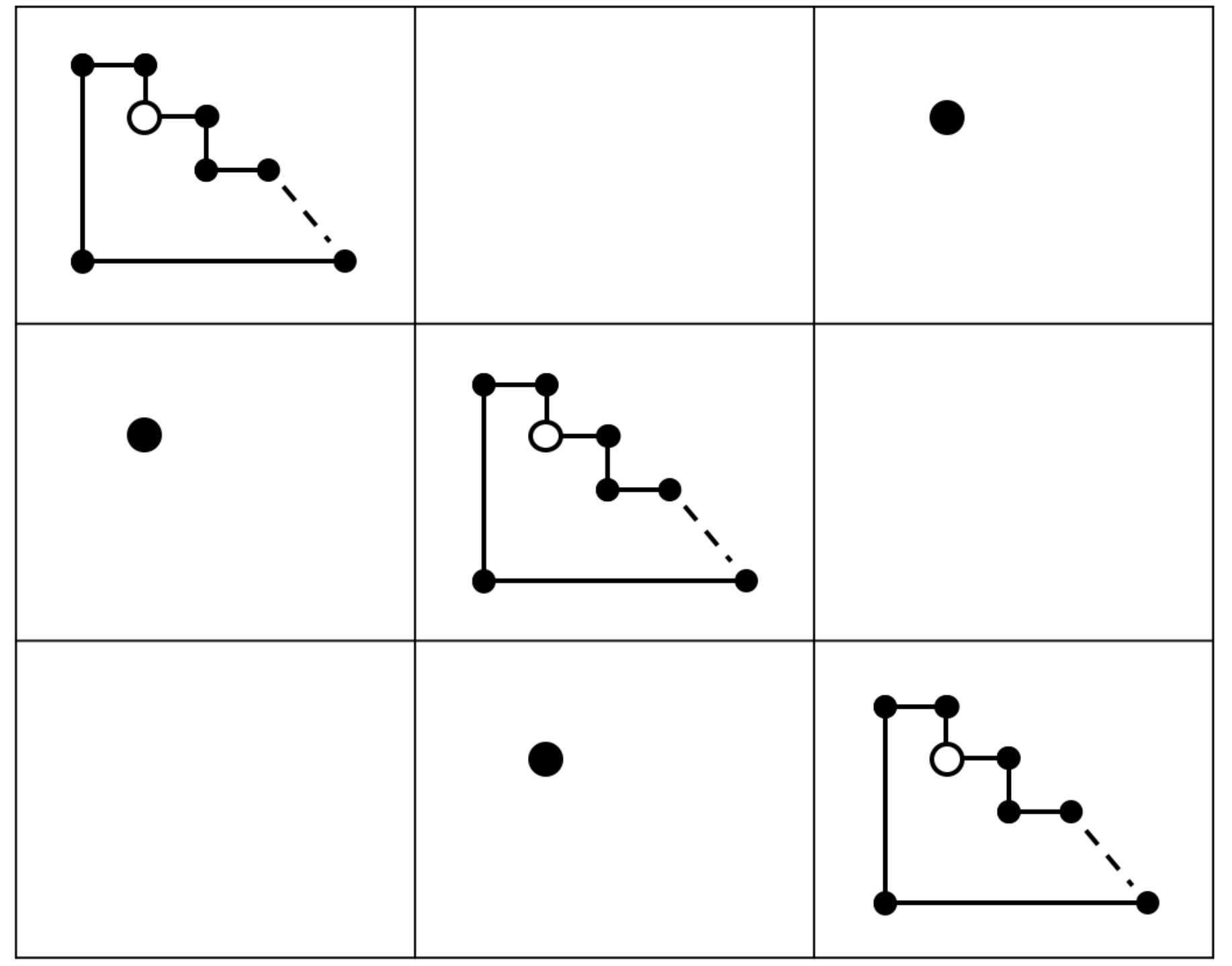}&\includegraphics[width=0.20\textwidth]{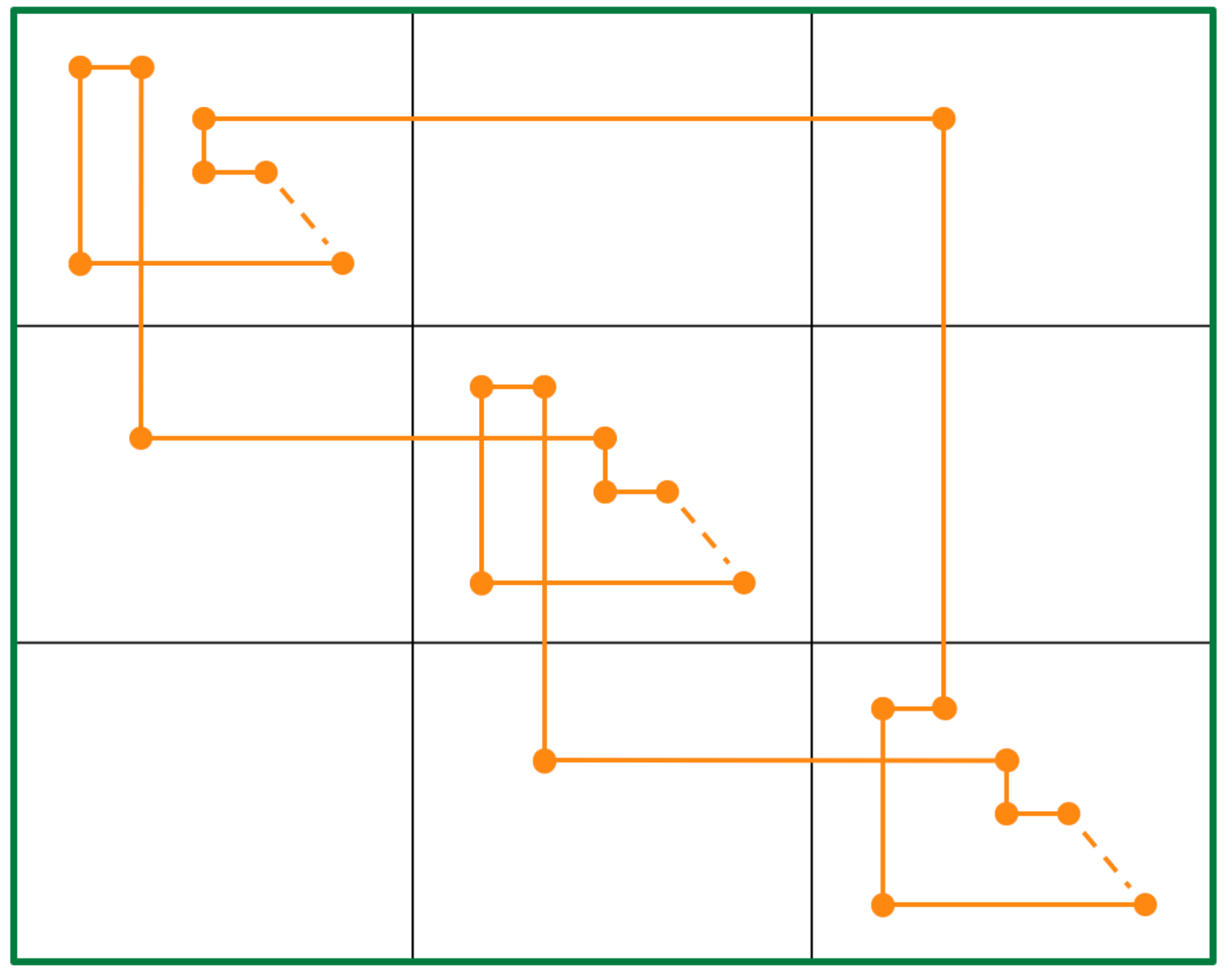}\\
(a)&(b)
\end{tabular}
\caption{(a) $\mathcal{C}_{i_a,j_a}{\rightarrow}\bold{P}$. The white circles show original locations of the relocated circulant. (b) A cycle-$3k$ is formed (effective relocation).\vspace{-0.3cm}}
\end{figure}
\begin{figure}
\centering
\begin{tabular}{cc}
\includegraphics[width=0.20\textwidth]{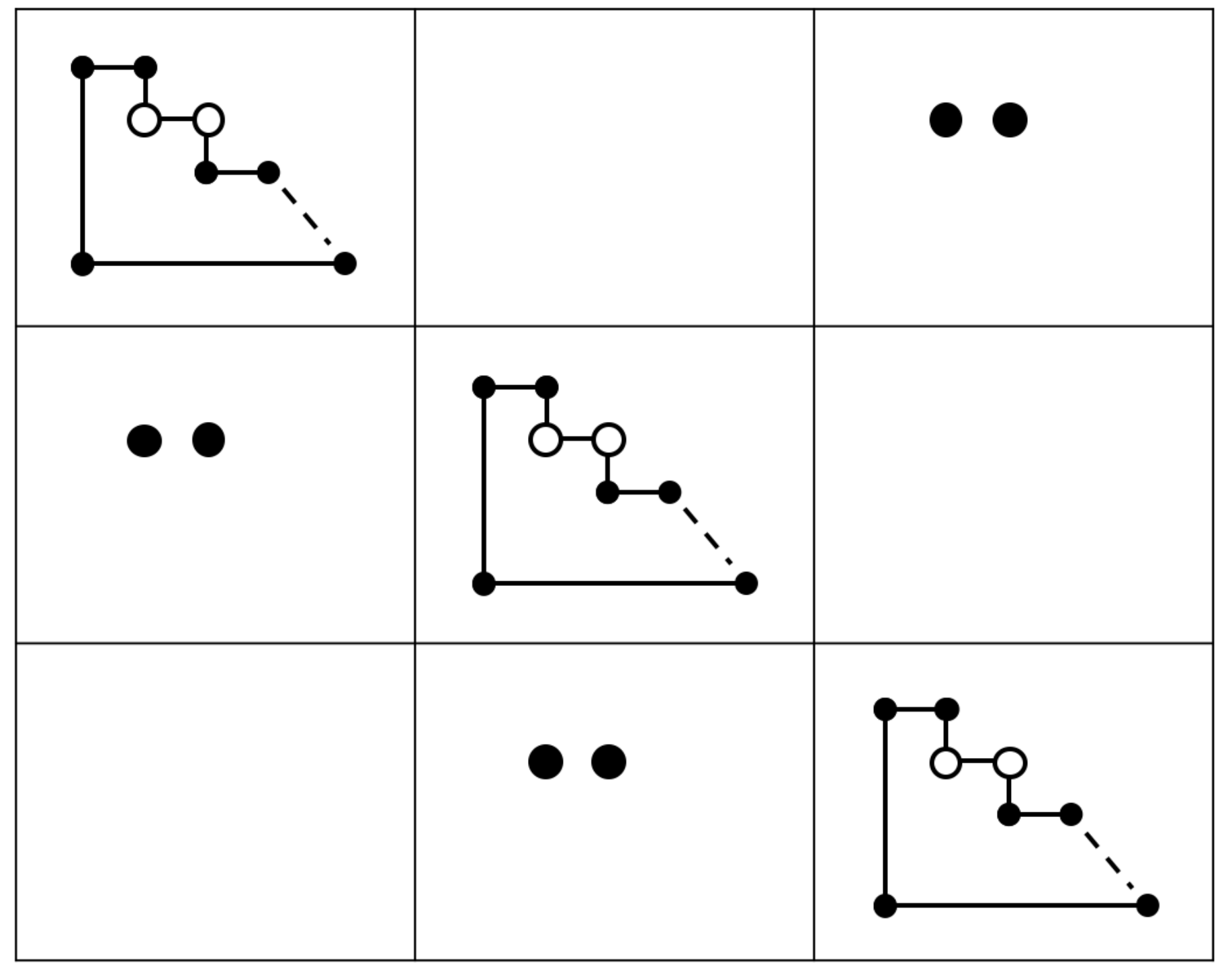}&\includegraphics[width=0.20\textwidth]{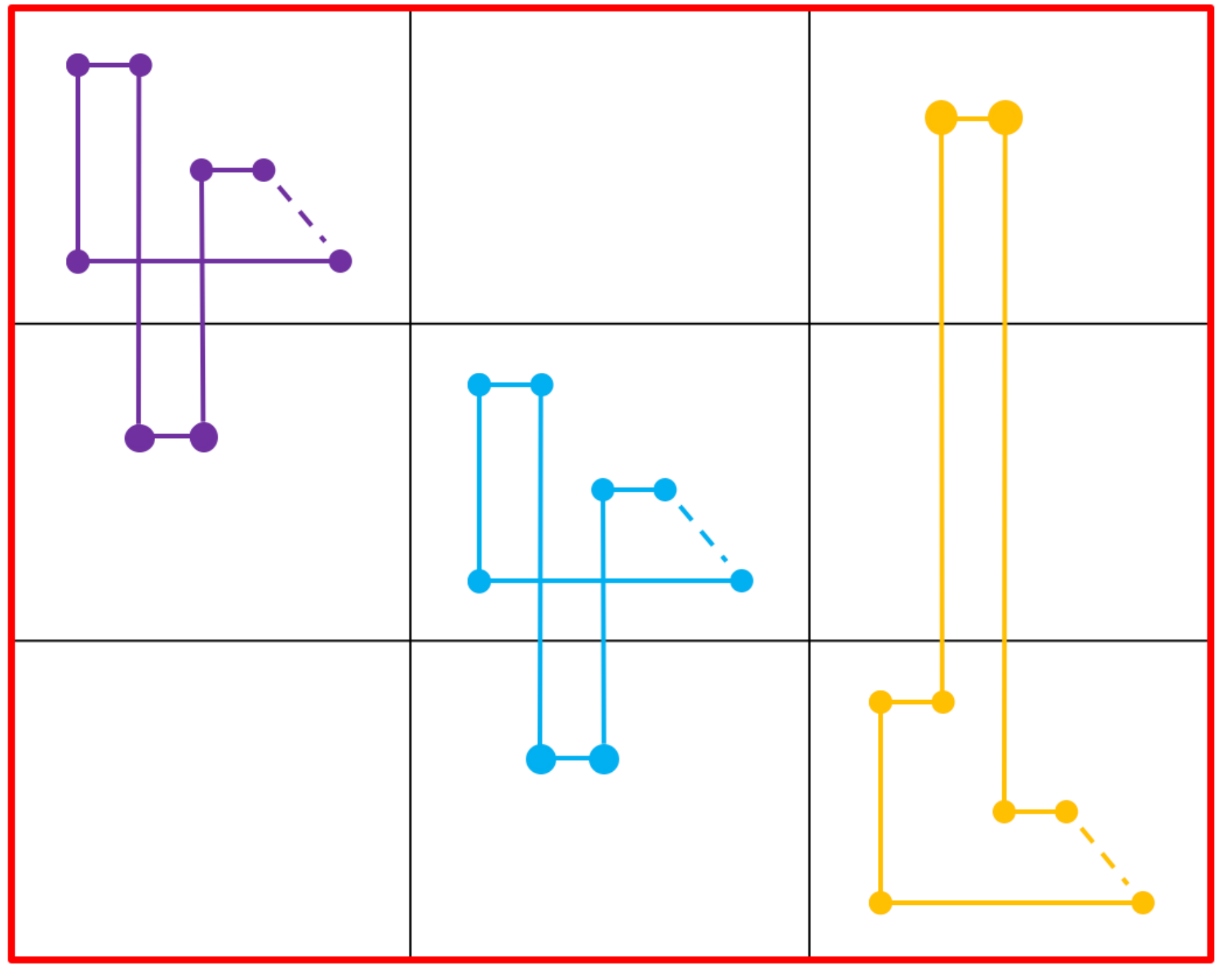}\\
(a)&(b)
\end{tabular}
\caption{(a) $\{\mathcal{C}_{i_a,j_a},\mathcal{C}_{i_b,j_b}\}{\rightarrow}\bold{P}$. (b) Three cycles-$k$ are formed (ineffective relocations).\vspace{-0.6cm}}
\end{figure}
\end{example}

\begin{remark}
A circulant can appear more than once in $C_{\mathcal{O}_k}$, e.g., see Fig.~3(b). A circulant that is repeated $r$ times in the sequence can be interpreted in our analysis as $r$ different circulants; {each two circulants from this group have an even distance on $\mathcal{O}_k$.} The relocation of a circulant that appears $r$ times {is then equivalent to the relocation of $r$ circulants with the above property to the same auxiliary matrix.}\vspace{-0.0cm}
\end{remark}

\subsection{Majority Voting Algorithm for MD-SC Code Design}
In this subsection, we explain our new systematic framework for constructing MD-SC codes. Our framework is based on a majority voting policy and aims at minimizing the population of cycles-$k$. As stated in Section~II, the MD coupling is performed via relocating problematic circulants to two auxiliary matrices, $\bold{P}$ and $\bold{Q}$. {After relocating one circulant, the ranking of the problematic circulants (with respect to the number of cycles each of them is involved in) changes. Thus, the relocations of circulants are performed sequentially.}

Initially, $\bold{H}_\text{SC}'=\bold{H}_\text{SC}$ and $\bold{P}=\bold{Q}=\bold{0}$. {At each iteration, one non-zero circulant in replica ${\bold{R}_d}$ of $\bold{H}_\text{SC}'$, {$d=\lceil L/2\rceil$}, which is the circulant involved in the {highest} number of cycles-$k$ in $\bold{H}_\text{SC}^\text{MD}$, is targeted for relocation.} Each cycle-$k$ in $\bold{H}_\text{SC}$ that has the targeted circulant in its sequence {gives votes (possibly zero votes)} regarding the relocation. The decision is made based on the majority of the votes. {The same decision also applies to the other $(L-1)$ instances of this circulant in $\bold{H}_\text{SC}'$.}
The relocations are performed sequentially until the MD coupling threshold is achieved, or relocation does not decrease the population of cycles-$k$ anymore.

Consider a targeted circulant $\mathcal{C}_{i_u,j_u}$. There are three possible actions for this circulant: {relocate to $\bold{P}$ ({i.e.,} $M(\mathcal{C}_{i_u,j_u})=1$), relocate to $\bold{Q}$ ({i.e.,} $M(\mathcal{C}_{i_u,j_u})=2$), and keep in $\bold{H}_\textnormal{SC}'$ ({i.e.,} $M(\mathcal{C}_{i_u,j_u})=0$)}. {Each cycle $\mathcal{O}_k$ in $\bold{H}_\text{SC}$ that has the targeted circulant in its sequence votes for a subset of these actions (which can be empty or can have more than one action). Cycles give their vote(s) on these actions as follows:} 
\begin{enumerate}
\item \textbf{Relocate to }$\bold{P}$: If $\mathcal{C}_{i_u,j_u}{\rightarrow}\bold{P}$ removes $\mathcal{O}_k$ from $\bold{H}_\text{SC}^\text{MD}$ ((4) is violated).
\item \textbf{Relocate to }$\bold{Q}$: If $\mathcal{C}_{i_u,j_u}{\rightarrow}\bold{Q}$ removes $\mathcal{O}_k$ from $\bold{H}_\text{SC}^\text{MD}$.
\item \textbf{Keep in }$\bold{H}_\textnormal{SC}'$: If leaving $\mathcal{C}_{i_u,j_u}$ in $\bold{H}_\textnormal{SC}'$ removes $\mathcal{O}_k$ from $\bold{H}_\text{SC}^\text{MD}$.
\end{enumerate}
We note that the votes are given considering the previously-relocated circulants of $\mathcal{O}_k$ and equation (4). Example~3 studies several scenarios, and in each one, $\mathcal{O}_k$ gives different votes regarding a targeted circulant.
\begin{example}
Consider cycle $\mathcal{O}_k$ and circulant $\mathcal{C}_{i_u,j_u}\in C_{\mathcal{O}_k}$.\\
\underline{Scenario 1:} No circulants of $\mathcal{O}_k$ are relocated in previous iterations, and $\mathcal{C}_{i_u,j_u}$ appears once in $C_{\mathcal{O}_k}$ (i.e., $r=1$). Thus, (4) is not satisfied {under relocation} regardless of which auxiliary matrix $\mathcal{C}_{i_u,j_u}$ is relocated to, and $\mathcal{O}_k$ votes for ``relocate to $\bold{P}$'' and ``relocate to $\bold{Q}$''.\\
\underline{Scenario 2:} No circulants of $\mathcal{O}_k$ are relocated in previous iterations, and $\mathcal{C}_{i_u,j_u}$ appears three times in $C_{\mathcal{O}_k}$ ($r=3$). Then, $\mathcal{O}_k$ cannot be removed via relocation, and $\mathcal{O}_k$ does not give any votes.
Refer to the relocation of three circulants with even  mutual distances to the same auxiliary matrix in Example~2.\\
\underline{Scenario 3:} Circulant $\mathcal{C}_{i_v,j_v}\in C_{\mathcal{O}_k}$ {is already relocated to $\bold{P}$ in previous iterations}, $D_{\mathcal{O}_k}(\mathcal{C}_{i_u,j_u},\mathcal{C}_{i_v,j_v})=2$, and both circulants appear once in $C_{\mathcal{O}_k}$ (i.e., $r=1$). Then, (4) is violated if $\mathcal{C}_{i_u,j_u}{\rightarrow}\bold{P}$ or with no relocation, and $\mathcal{O}_k$ votes for ``relocate to $\bold{P}$'' and ``keep in $\bold{H}_\textnormal{SC}'$''.
\end{example}

Assume there are $\beta$ cycles-$k$ in $\bold{H}_\textnormal{SC}$. Each of these $\beta$ cycles votes for a subset of actions $\{$relocate to $\bold{P}$, relocate to $\bold{Q}$, keep in $\bold{H}_\textnormal{SC}'\}$ 
for the targeted circulant. The collective voting result is considered for making a decision. {If the number of votes for ``keep in $\bold{H}_\textnormal{SC}'$'' is more than the number of votes for ``relocate to $\bold{P}$'' and ``relocate to $\bold{Q}$'', no relocation is performed. This is also an indicator for termination of the process since by relocating the most problematic circulant, the population of cycles-$k$ increases (the relocation does not help any more). Otherwise, based on which auxiliary matrix of the two, $\bold{P}$ and $\bold{Q}$, has more votes, the relocation is performed.}

We highlight three points here: 1) The targeted circulant $\mathcal{C}_{i_u,j_u}$ is chosen among the non-zero circulants of $\bold{H}_\text{SC}'$ to increase the multi-dimensional coupling. 2) The most problematic circulant is chosen based on the cycles-$k$ in $\bold{H}_\textnormal{SC}^\textnormal{MD}$ (active cycles). 3) Each cycle in $\bold{H}_\textnormal{SC}$ that visits {$\mathcal{C}_{i_u,j_u}$} (even those that are removed in $\bold{H}_\textnormal{SC}^\textnormal{MD}$ as a result of the previous relocations) gives votes regarding the relolation of {$\mathcal{C}_{i_u,j_u}$}. This is because some cycles that are removed in the previous iterations {may appear again after relocating new circulants}.

{As the final (optional) step}, a post-processing circulant power optimizer {(PP CPO)} is performed to remove as many as possible of remaining cycles-$k$ in $\bold{H}_\textnormal{SC}^\textnormal{MD}$. An algebraic condition must hold on the powers of a group of $k$ non-zero circulants in order that they form a cycle-$k$ in a CB code \cite{FossorierIT2004}. {The {PP CPO} adjusts the powers of the relocated circulants to break the necessary condition for as many remaining cycles-$k$ as possible. The power of a relocated circulant is changed if it results in a lower number of cycles-$k$ while creating zero cycles of lower lengths. Our {pseudo-algorithm} for constructing MD-SC codes is given as Algorithm~1.

\begin{algorithm}
\caption{{Constructing MD-SC Codes}}\label{algo1}
\begin{algorithmic}[1]
\State \textbf{Inputs:} $\bold{H}_\textnormal{SC}$, $k$, and $\mathcal{T}$.
\State \textbf{Initialize:} $\bold{P}=\bold{Q}=\bold{0}$, $\bold{H}_\textnormal{SC}'=\bold{H}_\textnormal{SC}$, and $M(\mathcal{C}_{i_u,j_u})=0$.
\State {Locate all cycles-$k$ in $\bold{H}_\textnormal{SC}$ that visit the circlants (at least two circulants) in {replica $\bold{R}_{\lceil L/2\rceil}$} of $\bold{H}_\textnormal{SC}$.}
\State Mark all cycles-$k$ found in step 3 as active.
\State Choose $\mathcal{C}_{i_u,j_u}\neq\bold{0}$ such that $M(\mathcal{C}_{i_u,j_u})=0$, and $\mathcal{C}_{i_u,j_u}$ {appears in sequences of active cycles more than any other circulant}.
\State Find all cycles (active/inactive) that have $\mathcal{C}_{i_u,j_u}$ in their sequences.
\State Each cycle in step 6 votes for a subset of actions $\{$relocate to $\bold{P}$, relocate to $\bold{Q}$, keep in $\bold{H}_\textnormal{SC}'\}$ that make it inactive.
\State {Find the action receiving the majority of votes.}
\State If the majority voting is for ``keep in $\bold{H}_\textnormal{SC}'$", go to step 13.
\State If the majority voting is for ``relocate to $\bold{P}$", $\mathcal{C}_{i_u,j_u}{\rightarrow}\bold{P}$. Otherwise, $\mathcal{C}_{i_u,j_u}{\rightarrow}\bold{Q}$.
\State Update the list of (active/inactive) cycles-$k$ using (4).
\State If the number of relocations is less than $\mathcal{T}$ and the number of active cycles is greater than $0$, go to step 5.
\State Construct $\bold{H}_\textnormal{SC}^\textnormal{MD}$ using (2).
\State {(Optional)} Perform {PP CPO} on the relocated circulants.
\State \textbf{Output:} $\bold{H}_\textnormal{SC}^\textnormal{MD}$.
\end{algorithmic}
\end{algorithm}

\section{Simulation Results}
\begin{figure*}[h!]
\centering
\footnotesize
\begin{tabular}{ll}
\hspace{-0.2cm}$\bold{PM}^1\hspace{-0.1cm}=\hspace{-0.1cm}\left[
\arraycolsep=3pt\def\arraystretch{1}
\begin{array}{ccccccccccccccccc}
0&1&0&1&0&1&0&1&0&1&0&1&0&1&0&1&1\\
1&0&1&0&1&0&1&0&1&0&1&0&1&0&1&0&0\\
0&0&0&0&0&0&0&0&1&1&1&1&1&1&1&1&1\\
1&1&1&1&1&1&1&1&0&0&0&0&0&0&0&0&0
\end{array}\right]$\hspace{-0.05cm},&\hspace{-0.5cm}
$\bold{CM}^1\hspace{-0.1cm}=\hspace{-0.1cm}\left[
\arraycolsep=2pt\def\arraystretch{1}
\begin{array}{ccccccccccccccccc}
0 &10&2&8&2&0&5&7&15&0&0&0&0&10&0&0&0\\
11&15&2&14&10&3&6&7&8&9&4&11&12&8&14&10&16\\
11&2&4&12&8&11&12&9&15&4&13&5&6&1&11&13&15\\
11&3&6&9&2&16&8&4&7&10&13&16&2&5&8&6&14
\end{array}\right]$.\\
\multicolumn{2}{c}{(a)}\\
\hspace{-0.2cm}$\bold{PM}^2\hspace{-0.1cm}=\hspace{-0.1cm}\left[
\arraycolsep=3pt\def\arraystretch{1}
\begin{array}{ccccccccccccccccccc}
0&1&1&0&1&2&0&2&2&0&1&1&0&1&2&0&2&2&2\\
1&0&0&1&0&0&1&0&0&2&2&2&2&2&1&2&1&1&1\\
2&2&2&2&2&1&2&1&1&1&0&0&1&0&0&1&0&0&0
\end{array}\right]$\hspace{-0.05cm},&\hspace{-0.47cm}
$\bold{CM}^2\hspace{-0.1cm}=\hspace{-0.1cm}\left[
\arraycolsep=2pt\def\arraystretch{1}
\begin{array}{ccccccccccccccccccc}
21&0 &16&3&19&1&0&0&21&5&0&0&1&0&9&0&16&1&0\\
 0 &11&7&3&4&5&6&7&8&9&10&11&12&13&14&15&16&17&18\\
 0&17&0 &6&8&10&12&14&16&18&20&22&1&3&5&19&9&11&13
 \end{array}\right]$\hspace{-0.05cm}.\\
\multicolumn{2}{c}{(b)}
\end{tabular}
\caption{The partitioning matrices and circulant power matrices obtained by the OO-CPO technique: a) SC-Code-1 {and 3}. b) SC-Code-2 {and 4}.}
\end{figure*}

\begin{figure*}[h!]
\footnotesize
\begin{tabular}{ll}
\hspace{-0.2cm}$\bold{M}^1\hspace{-0.0cm}=\hspace{-0.0cm}\left[
\arraycolsep=3pt\def\arraystretch{1}
\begin{array}{ccccccccccccccccc}
0&1&2&0&1&1&0&0&0&0&0&0&0&0&0&0&0\\
1&0&1&0&0&0&0&1&0&0&0&0&0&0&0&0&0\\
0&0&0&0&2&0&0&1&1&1&0&0&1&0&2&0&1\\
0&0&0&0&0&0&0&0&0&0&0&0&0&1&0&0&0
\end{array}\right]$,&\hspace{-0.3cm}
$\bold{M}^2\hspace{-0.0cm}=\hspace{-0.0cm}\left[
\arraycolsep=3pt\def\arraystretch{1}
\begin{array}{ccccccccccccccccccc}
0&1&0&0&0&0&0&0&0&0&0&0&0&2&0&0&0&0&0\\
0&0&0&1&0&0&1&0&0&0&0&0&0&0&1&0&1&1&0\\
0&0&0&0&0&2&0&0&0&1&0&1&2&0&0&0&0&0&1
\end{array}\right]$.\\
\multicolumn{2}{c}{(a)}\\ \multicolumn{2}{c}{
$\bold{CM}^3\hspace{-0.0cm}=\hspace{-0.0cm}\left[
\arraycolsep=2pt\def\arraystretch{1}
\begin{array}{ccccccccccccccccc}
0&\bold{10}&\bold{9}&8&\bold{2}&\bold{0}&5&7&15&0&0&0&0&10&0&0&0\\
\bold{14}&15&\bold{2}&14&10&3&6&\bold{12}&8&9&4&11&12&8&14&10&16\\
11&2&4&12&\bold{8}&11&12&\bold{14}&\bold{15}&\bold{4}&13&5&\bold{14}&1&\bold{11}&13&\bold{15}\\
11&3&6&9&2&16&8&4&7&10&13&16&2&\bold{5}&8&6&14
\end{array}\right]$.}\\
 \multicolumn{2}{c}{(b)}\\
\end{tabular}
\caption{(a) The MD mapping matrices for MD-SC-Code-1 and 2. (b) Circulant power matrices for MD-SC-Code-3.}
\end{figure*}

In this section, we compare MD-SC codes constructed by our new framework with their 1D-SC counterparts (1D-SC codes having the same length and nearly the same rate as the MD-SC codes). First, we demonstrate the reduction in the number of cycles-$k$ achieved by the effective MD coupling, {where $k$ is the length of the shortest cycle in the graph of the code.} Then, we show the performance improvement via BER curves. In our simulations, we consider the AWGN channel, and SC codes with {parameters} $m\in\{1,2\}$ and $\gamma\in\{3,4\}$.

Our code parameters are as follows: SC-Code-1 is an SC code with $\kappa=z=17$, $\gamma=4$, $m=1$, $L=10$, rate $0.74$, and length $2{,}890$ bits. SC-Code-2 is an SC code with $\kappa=19$, $z=23$, $\gamma=3$, $m=2$, $L=10$, rate $0.81$, and length $4{,}370$ bits. SC-Code-1 and SC-Code-2 are constructed by the OO-CPO technique \cite{1802_06481}. The partitioning matrix $\bold{PM}=[h_{i,j}]$ and circulant power matrix $\bold{CM}=[f_{i,j}]$, with dimensions $\gamma\times\kappa$, describe partitioning and circulant powers{, respectively}. A circulant with row group index $i$ and column group index $j$ in the block code $\bold{H}$ is assigned to the component matrix $\bold{H}_{h_{i,j}}$, and it has power $f_{i,j}$. These two matrices for SC-Code-1 and SC-Code-2 are shown in Fig.~6. {SC-Code-3 is similar to SC-Code-1 but with $L=30$. Thus, its rate and length are $0.76$ and $8{,}670$ bits, respectively. SC-Code-4 is similar to SC-Code-2 but with $L=30$. Thus, its rate and length are $0.83$ and $13{,}110$ bits, respectively. SC-Code-1 and SC-Code-3 have girth $6$. SC-Code-2 and SC-Code-4 have girth $8$.}

MD-SC-Code-1 and MD-SC-Code-2 are MD-SC codes constructed by our new framework. Their constituent SC codes are SC-Code-1 and SC-Code-2, and their MD densities are $\mathcal{T}_1=15$ (i.e., $22.06\%$) and $\mathcal{T}_2=12$ (i.e., $21.05\%$), respectively. The cycles of interest that we aim at minimizing their population in the MD code design are cycles-$6$ for MD-SC-Code-1 and cycles-$8$ for MD-SC-Code-2. The MD mapping matrices, i.e., $\bold{M}=[M(\mathcal{C}_{i,j})]$, obtained by our new framework are shown in Fig.~7(a). 
We note that $\mathcal{T}_1$ and $\mathcal{T}_2$ are the maximum relocations needed for constructing MD-SC codes. In particular, after reaching these two thresholds, relocating more problematic circulants does not decrease the population of the cycle of interest.
{MD-SC-Code-3 is constructed from MD-SC-Code-2 by applying the optional PP CPO step (step 14 of Algorithm~1), and its circulant power matrix is shown in Fig.~7(b). MD-SC-Code-1 has rate $0.74$ and length $8{,}670$ bits. MD-SC-Code-2 and MD-SC-Code-3 have rate $0.81$ and length $13{,}110$ bits.}

{TABLE~I shows the number of cycles of interest for MD-SC codes constructed by our new framework compared to their 1D-SC counterparts. According to TABLE~I, MD-SC-Code-1 has nearly $84\%$ fewer cycles-$6$ compared to SC-Code-3. Moreover, MD-SC-Code-2 and MD-SC-Code-3 have nearly $73\%$ and $75\%$ fewer cycles-$8$ compared to SC-Code-4, respectively.} 

\begin{figure}
	\vspace{-0.3cm}
	\centering
	\includegraphics[width=0.42\textwidth]{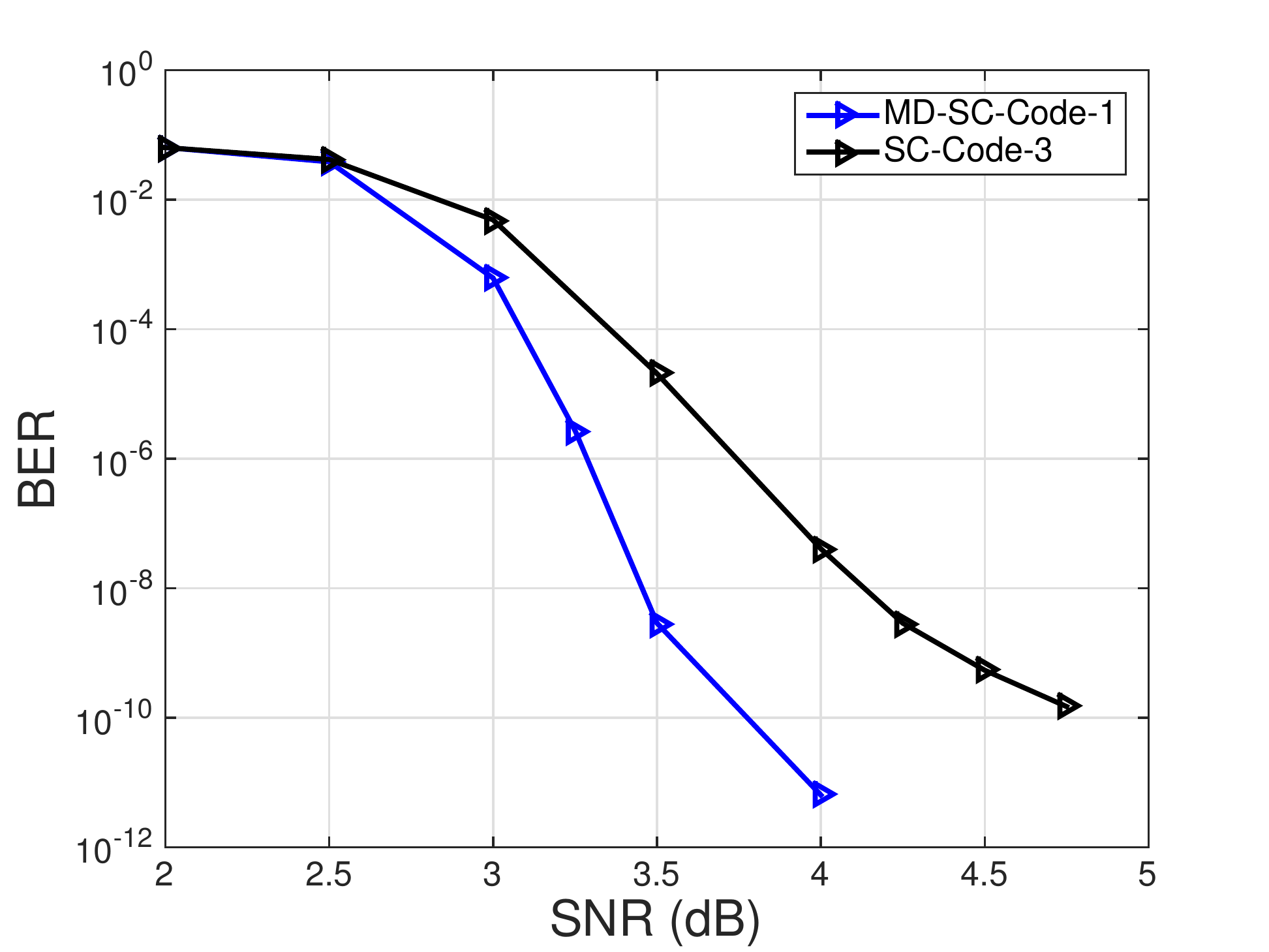}
	\vspace{-0.0cm}
	\caption{{The performance comparison of MD-SC-Code-1 and SC-Code-3.}}
	\vspace{-0.0cm}
\end{figure}

\begin{figure}
	\centering
	\vspace{-0.3cm}
	\includegraphics[width=0.42\textwidth]{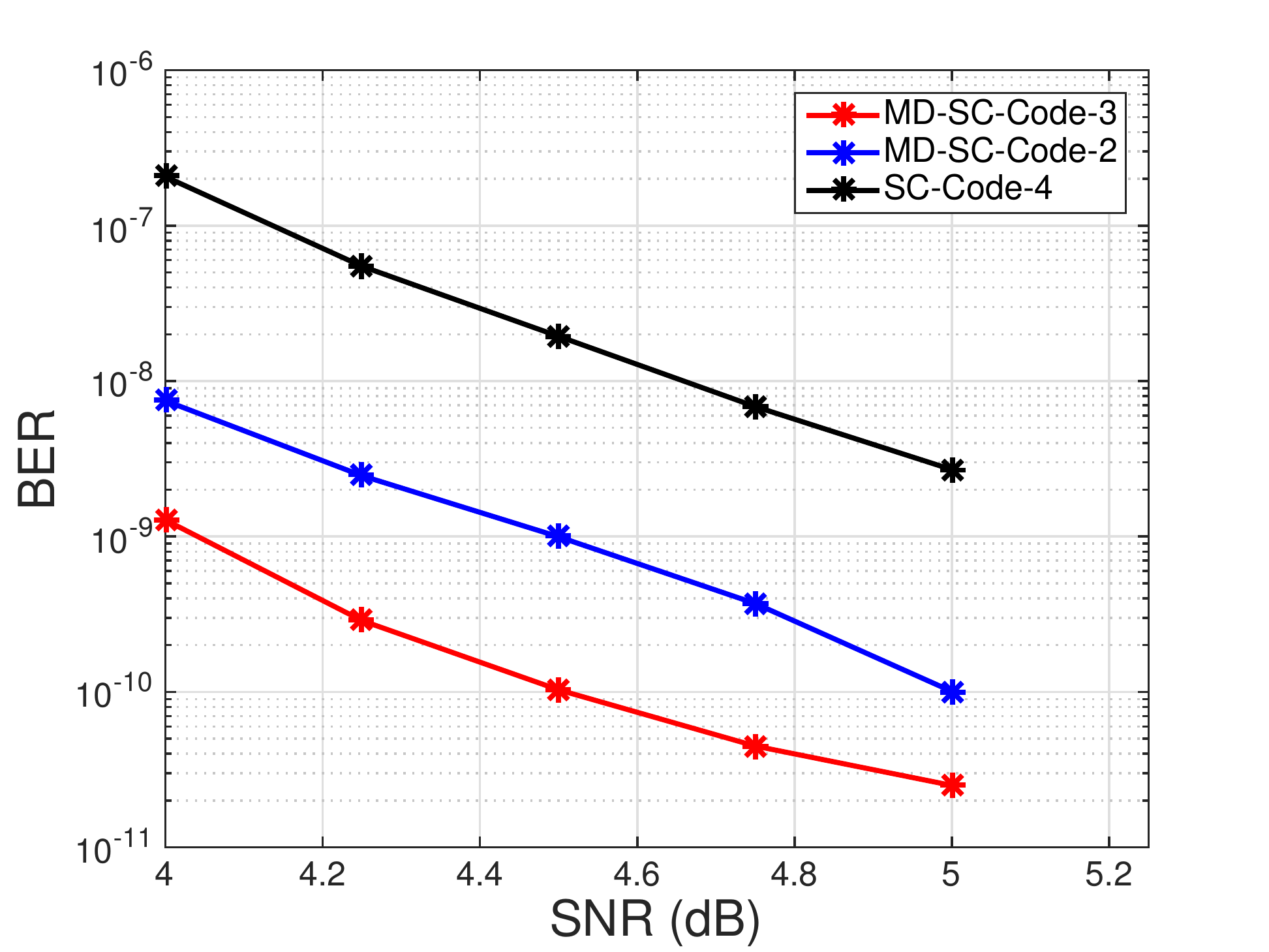}
	\vspace{-0.0cm}
	\caption{{The performance comparison of MD-SC-Code-2, MD-SC-Code-3, and SC-Code-4.}}
	\vspace{-0.0cm}
\end{figure}

\begin{table}
\centering
\caption{Comparison of the population of cycles of interest for MD-SC codes with their {1D-SC counterparts}.}
\begin{tabular}{|c|c|c|}
\hline
&{SC-Codes-3}&MD-SC-Code-1\\
\hline
cycles-$6$&{$91{,}494$}&$14{,}331$\\
\hline
\end{tabular}\vspace{0.1cm}

\begin{tabular}{|c|c|c|c|}
\hline
&{SC-Codes-4}&MD-SC-Code-2&MD-SC-Code-3\\
\hline
cycles-$8$&{$1{,}034{,}609$}&$280{,}968$&$253{,}851$\\
\hline
\end{tabular}\vspace{-0.35cm}
\end{table}

{Finally, we compare the performance of our MD-SC codes with their 1D-SC counterparts. 
As Fig.~8 shows, MD coupling results in nearly $3.8$ orders of magnitude BER performance improvement for MD-SC-Code-1 compared to its 1D-SC counterpart, i.e., SC-Code-3, at SNR$=4.0$ dB. As Fig.~9 shows, MD coupling (resp., MD coupling along with PP CPO) results in nearly {$1.3$} (resp., {$2.2$}) orders of magnitude BER performance improvement for MD-SC-Code-2 (resp., MD-SC-Code-3) compared to its 1D-SC counterpart, i.e., SC-Code-4, at SNR$=4.75$ dB.}

\section{Conclusion}
We expanded the repertoire of SC codes by establishing a framework for MD-SC code construction. For MD coupling, we rewire connections that are the most problematic within each SC code. Our framework encompasses a systematic way to sequentially identify and relocate problematic circulants, thus utilizing them to connect the constituent SC codes. Our MD-SC codes show a notable reduction in population of the smallest cycles and a significant improvement in BER performance compared to the 1D setting.
Two promising research directions are to investigate MD-SC codes on non-uniform channels, such as multilevel Flash and multi-dimensional magnetic recording channels, in addition to presenting low-latency decoding by incorporating code locality into decoder implementations.

\section*{Acknowledgment}
Research supported in part by a grant from ASTC-IDEMA.
\bibliographystyle{IEEEtran}
\bibliography{IEEEabrv,references}

\end{document}